\documentclass[aps,prl,twocolumn,showpacs]{revtex4-2}
\usepackage[english]{babel}
\usepackage[utf8]{inputenc}
\usepackage[colorlinks=true,citecolor=blue]{hyperref}
\usepackage{physics}
\usepackage{amssymb,amsmath,amsbsy,bm,amscd}
\usepackage{amsthm}
\usepackage[shortlabels,inline]{enumitem}
\usepackage{epsf,epsfig}
\usepackage{graphicx}
\usepackage{algorithm}
\usepackage{verbatim}
\usepackage{algpseudocode}
\usepackage{url}
\usepackage{natbib}
\usepackage{hyperref}
\hypersetup{
    colorlinks,
    linkcolor={red},
    citecolor={blue},
    urlcolor={blue}
}

\begin{document}
\title{Frozen-Tree Sampling Refutes Quantum Advantage of Random Circuit Sampling}
\author{Sangchul Oh}
\email[]{sangchul.oh@siu.edu}
\affiliation{School of Physics and Applied Physics, Southern Illinois University Carbondale, 
IL 62906, USA}
\date{\today}

\begin{abstract}
Random circuit sampling of bitstrings from a Haar-random quantum state is widely believed to be 
classically intractable, and has therefore been implemented as a primary benchmark for 
demonstrating quantum advantage. Here, we challenge this premise by proposing an efficient 
classical frozen-tree sampling algorithm that exploits the conditional scale invariance of 
Haar-random quantum states [Oh, arXiv:2602.19448]. The frozen-tree sampler draws bitstrings of $n$
qubits in $O(n)$ time per sample. Moreover, its output probability $p_F(x)$ is statistically 
identical to the probability $p_C(x)$ of a random quantum circuit, since both are independent 
instances of the same Dirichlet distribution. Consequently, no statistical test acting on samples 
alone can distinguish the classical frozen-tree sampler from a quantum random circuit. 
The claimed quantum advantage of random circuit sampling therefore does not withstand scrutiny: 
its hardness lies not in sampling from the Dirichlet distribution, which is classically efficient, 
but in identifying a specific circuit realization.
\end{abstract}
\maketitle

\paragraph*{Introduction}--- Quantum advantage, the outperformance of quantum computers over 
classical digital computers on certain tasks, is considered one of the most important milestones 
in quantum computation. Random circuit sampling (RCS) is regarded as a primary benchmark for
demonstrating quantum advantage on current noisy intermediate-scale quantum computers. Claims of 
quantum advantage in RCS have recently been reported using superconducting 
qubits~\cite{Arute2019,Wu2021,Zhu2022,Morvan2024,Gao2025} and ion-trap 
qubits~\cite{Liu2025,DeCross2025,Ransford2026}. Operationally, RCS is the task of sampling bitstrings 
from a Haar-random quantum state, hereafter referred to as a random quantum state, generated 
by a random quantum circuit~\cite{Boixo2018}. RCS is believed to be classically intractable because 
a random quantum state is highly entangled~\cite{Page1993a} and appears too chaotic for a classical 
algorithm to exploit any pattern or structure~\cite{Bouland2019,Aaronson2020,Hangleiter2023,Movassagh2023}. 
Statistical properties of output bitstrings such as the exponential distribution, the linear cross-entropy 
benchmark~\cite{Boixo2018,Arute2019,Dalzell2024,Gao2024}, heavy output generation~\cite{Aaronson2017} 
and anti-concentration~\cite{Aaronson2013,Bremner2016,Hangleiter2018} 
have been proposed as evidence for the quantum advantage of RCS.

In this paper, we challenge the premise of quantum advantage in RCS by introducing frozen-tree 
sampling, which exploits the exact conditional scale invariance of a random quantum 
state~\cite{Oh2026,Oh2026b}, as shown in Fig.~\ref{Fig:1}. We prove that the probability of 
finding bitstrings of a random quantum state can be represented by a binary tree 
with a precise recursive structure: each bit is drawn from a Beta-distributed conditional probability 
determined by the preceding bits. We present a frozen-tree sampler that classically samples bitstrings 
from an $n$-qubit random quantum state in $O(n)$ time per sample. Both the probability $p_C(x)$ 
of a random quantum circuit and the probability $p_F(x)$ of the frozen-tree sampler are independent 
realizations of the Dirichlet vector characterizing a random quantum state, and are therefore 
statistically identical. It follows that no statistical verification method can serve as evidence of 
quantum advantage in RCS.

\begin{figure}[t]
\includegraphics[width=0.45\textwidth]{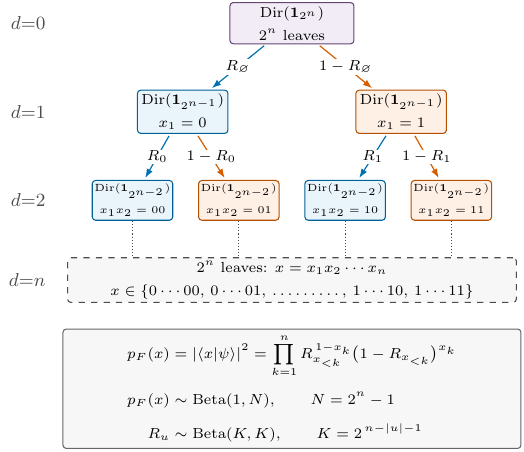}
\caption{Binary-tree representation of the probability $p(x)=\abs{\braket{x}{\psi}}^2
=\prod_k R_{x<k}^{1-x_k} (1-R_{x<k})^{x_k}$ of finding a bitstring $x = x_1\cdots x_n$ for a random quantum 
state $\ket{\psi}$ of $n$ qubits. The leaf probability $p(x)$ is uniquely determined by the product of branch 
ratios $R_{u}$ at depth $d=|u|=k-1$ with prefix $u=x_{<k}= x_1\cdots x_{k-1}$, where $R_u\sim\text{Beta}(K,K)$ 
with $K=2^{n-|u|-1}$. Each subtree is conditionally scale invariant and statistically identical to 
the whole system~\cite{Oh2026}.
\label{Fig:1}}
\end{figure}

\paragraph*{Dirichlet distribution of a Haar-random quantum state}---
A random pure quantum state of $n$ qubits in an $N=2^n$ dimensional Hilbert space~\cite{Wootters1990} 
is written as
\begin{align}
\ket{\psi} = \sum_{x=0}^{N-1} c_x \ket{x}\,,
\end{align}
where the bitstring $x=x_1\cdots x_n$ with $x_i\in \{0,1\}$ runs from $0$ to $N-1$ and each 
amplitude $c_x$ is sampled from the complex normal distribution, $c_x \sim {\cal CN}(0,1)$.  
The probability of finding bitstring $x$ is given by 
$p(x) = \abs{c_x}^2/\sum_y \abs{c_y}^2$. The probability vector $\boldsymbol{p}= (p(x))_{x=0,\dots N-1}$ 
follows a flat Dirichlet distribution on the $N-1$ simplex~\cite{Bengtsson2007}
\begin{align}
\left(p(x)\right) \sim \text{Dir}(1,\dots,1)_{2^n}\,.
\end{align}
The marginal distribution $p(x)$ follows a Beta distribution, $p\sim \text{Beta}(1,N-1)$, which becomes 
an exponential distribution in the rescaled random variable $Np$ in the limit of large 
$N$~\cite{Kus1988,Haake1990}. A quantum computer realizes a random quantum state by applying a circuit $C$,
chosen randomly according to the Haar measure, to an input state, $\ket{\psi} = C\ket{0^{n}}$, and performs 
a measurement in the computational basis to sample a bitstring $x$ from the probability distribution 
$p_C(x) = \abs{\bra{x}C\ket{0^n}}^2=\abs{c_x}^2/\sum_y|c_y|^2$.  Note that a random circuit $C$ yields a single 
realization of the Dirichlet distribution, $\boldsymbol{p}_C = (p_C(x)) \sim \text{Dir}(1,\dots,1)$. 

\paragraph*{Binary-tree representation of $p(x)$}---
The output bitstrings of RCS appear so random that no pattern or correlation is readily discernible.
Indeed, they pass the NIST statistical test for randomness~\cite{NIST2010,Oh2022a,Oh2022b,Oh2023}. Consider, 
for example, two $n=5$ bitstrings, $01101$ and $01100$, obtained from RCS. Both share 
the same four-bit prefix $u=x_1x_2x_3x_4 = 0110$, but differ in the final bit, $x_5 \in \{0,1\}$. 
The central question of this paper is: \emph{given the prefix $u=x_1x_2x_3x_4=0110$, what is the probability 
of obtaining $x_5 =0$ or $x_5=1$}? To answer it, we exploit the conditional scale invariance of a random 
quantum state established in Refs.~\cite{Oh2026,Oh2026b}. A random quantum state is not disordered 
but conceals an exact self-similar symmetry, as shown in Fig.~\ref{Fig:1}. This makes it possible
to express $p(x)$ as a binary tree with an analytic split probability at each depth and 
to sample bitstrings classically in $O(n)$ per sample.

\newtheorem{theorem}{Theorem}
\begin{theorem}[Tree Representation] \label{Frozen_tree}
For a random pure quantum state $\ket{\psi}$ of $n$ qubits, the probability 
$p(x) =\abs{\braket{x}{\psi}}^2$ of finding bitstring $x=x_1x_2\cdots x_n$ can be written as
the recursive tree form
\begin{align}
p(x) = \prod_{k=1}^{n} R_{x<k}^{\,1-x_{k}}\, \left(1 -R_{x<k}\right)^{x_{k}} \,,
\label{Eq:recursive}
\end{align}
where the branch ratios $\{R_u\}$ are mutually independent, and each branch ratio $R_u \in [0,1]$
between two equal child subtrees at node $u$ follows the Beta distribution 
\begin{align}
R_u \sim \text{\rm Beta}(K,K)\,.
\label{Eq:Branch_ratio}
\end{align}
Here, $u \equiv x_{<k} = x_1x_2\cdots x_{k-1}$ is the prefix with depth $d=k-1$, and 
$K=2^{n-|u|-1}$ is the size of each child subtree.
\end{theorem}

\begin{proof}[Proof]
The recursive tree expression for $p(x)$, Eqs.~(\ref{Eq:recursive}) and (\ref{Eq:Branch_ratio}), 
follows from three ingredients: the chain rule of probability, the aggregation property 
of the Dirichlet distribution~\cite{Lukacs1955,Kotz2019}, and the exact conditional scale invariance 
of the Dirichlet distribution~\cite{Oh2026}.

First, the chain rule expresses the joint probability $p(x)$ as a product of 
conditionals,
\begin{align}
p(x_1,x_2,\dots,x_n) = \prod_{k=1}^n p(x_k\mid x_1,\dots,x_{k-1})\,.
\end{align}
For every prefix $u \equiv x_{<k}=x_1\cdots x_{k-1}\in \{0,1\}^{k-1}$ at tree depth $d=|u|=k-1$, 
define the branch ratio 
\begin{align}
R_{u} \equiv p(x_k=0\mid u)\,.
\end{align}
The next bit $x_k$ is 0 with probability $R_u$ (left branch) and 1 with probability $(1-R_{u})$ 
(right branch), so the conditional probability reads
\begin{subequations}
\begin{align}
p(x_k\mid x_1,\dots,x_{k-1}) 
&=\begin{cases}
  R_{x<k}\,,     & x_k = 0 \\
  1 -R_{x<k}\,, & x_k =1
\end{cases} \\[6pt]
&= R_{x<k}^{\,1-x_k}\, (1-R_{x<k})^{x_k}\,.
\end{align}
\end{subequations}
Substituting into the chain rule yields Eq.~(\ref{Eq:recursive}).

Second, the exact conditional scale invariance of a random state that we established in Ref.~\cite{Oh2026} 
states that, writing $x=uz$, the conditional probability vector $\bigl(p(z\mid u)\bigr)_z$ again follows a flat Dirichlet 
distribution,
\begin{align}
\bigl(p(z\mid u)\bigr) \sim \mathrm{Dir}(1,\dots,1)_{2^{n-|u|}}\,.
\end{align}
Split the conditional subtree into two halves, $z=0w$ and $z=1w$,
\begin{subequations}
\label{Eq:Split_ratio}
\begin{align}
R_u &= p(0\mid u) = \sum_w p(0w\mid u) \,,\\
1-R_u &= p(1\mid u) = \sum_w p(1w\mid u)\,.
\end{align}
\end{subequations}

Third, since $\bigl(p(z\mid u)\bigr) \sim \mathrm{Dir}(1,\dots,1)$ and each half aggregates 
$K=2^{n-|u|-1}$ leaves, the aggregation property of the Dirichlet distribution~\cite{Lukacs1955,Kotz2019}
gives
\begin{align}
\Bigl( \sum_w p(0w\mid u),\ \sum_w p(1w\mid u) \Bigr) \sim \mathrm{Dir}(K,K).
\end{align}
A two-component Dirichlet is a Beta distribution, so
\begin{align}
R_u \sim \mathrm{Beta}(K,K)\,,\qquad K=2^{n-|u|-1}\,.
\end{align}
Finally, the branch ratios at distinct nodes are mutually independent. By Dirichlet neutrality, 
the split fraction at $u$ is independent of the conditional vector on each child subtree, and iterating 
down the tree makes all $\{R_u\}$ independent. 
This establishes Eqs.~(\ref{Eq:recursive}) and (\ref{Eq:Branch_ratio}).
\end{proof}

As shown in Figs.~\ref{Fig:1} and ~\ref{Fig:2}, Theorem~\ref{Frozen_tree} represents the probability $p(x)$ 
of finding bitstring $x$ for a random state as a binary tree walk from the root to leaves. The universal 
statistics of RCS is the Beta distribution. The full leaf probability $p(x)$ 
follows $\text{Beta}(1,N-1)$, which becomes exponential in the limit of large $N$. More generally, 
for any prefix $u$ (with $x=uw$), the subtree mass $P_u\equiv\sum_w p(uw)$, the weight of one subtree 
relative to the whole tree, follows $P_u\sim\mathrm{Beta}(S_d,N-S_d)$, where $S_d=2^{n-d}$ is
the number of leaves in the subtree at depth $d$~\cite{Oh2026}. For an internal node $u$
with children $u0$ and $u1$, the mass splits as $P_u=P_{u0}+P_{u1}$, and the pair satisfies 
$(P_{u0},P_{u1})\sim\mathrm{Dir}(K,K)$ with $K=S_2/2 = 2^{n-d-1}$. The split ratio in Eq.~(\ref{Eq:Split_ratio}) 
can then be written as $R_u \equiv {P_{u0}}/{P_u} \sim \text{Beta}(K,K)$,
so that $P_{u0} = R_uP_u$ and $P_{u1} = (1-R_u)P_u$.
For $R_u\sim \mathrm{Beta}(K,K)$, the mean and standard deviation are
\begin{align}
\mathbb{E}[R_u] =\frac{1}{2} \,,\quad
\sigma_d =\frac{1}{2\sqrt{2K+1}} \,,
\label{Variance}
\end{align}
Eq.~(\ref{Variance}) shows that the fluctuation of the split ratio depends only on 
the depth $d$ through $K=2^{n-d-1}$. Near the root $(K\gg 1)$ the branching ratio concentrates at $1/2$ 
with vanishing spread, $\sigma_d\to 0$, whereas near leaves $(K\approx 1)$ it is nearly uniform on 
$[0,1]$. Because $\sigma_d$ is a universal, depth-dependent quantity, it distinguishes ideal RCS from noisy RCS. 
As discussed below, we define the branch-ratio fidelity
$ F_\sigma \equiv \sigma_d^{\rm sample}/\sigma_d^{\rm ideal}$
as the ratio of the empirical branch fluctuation of the samples to its ideal value.

\begin{figure}[t]
\includegraphics[width=0.45\textwidth]{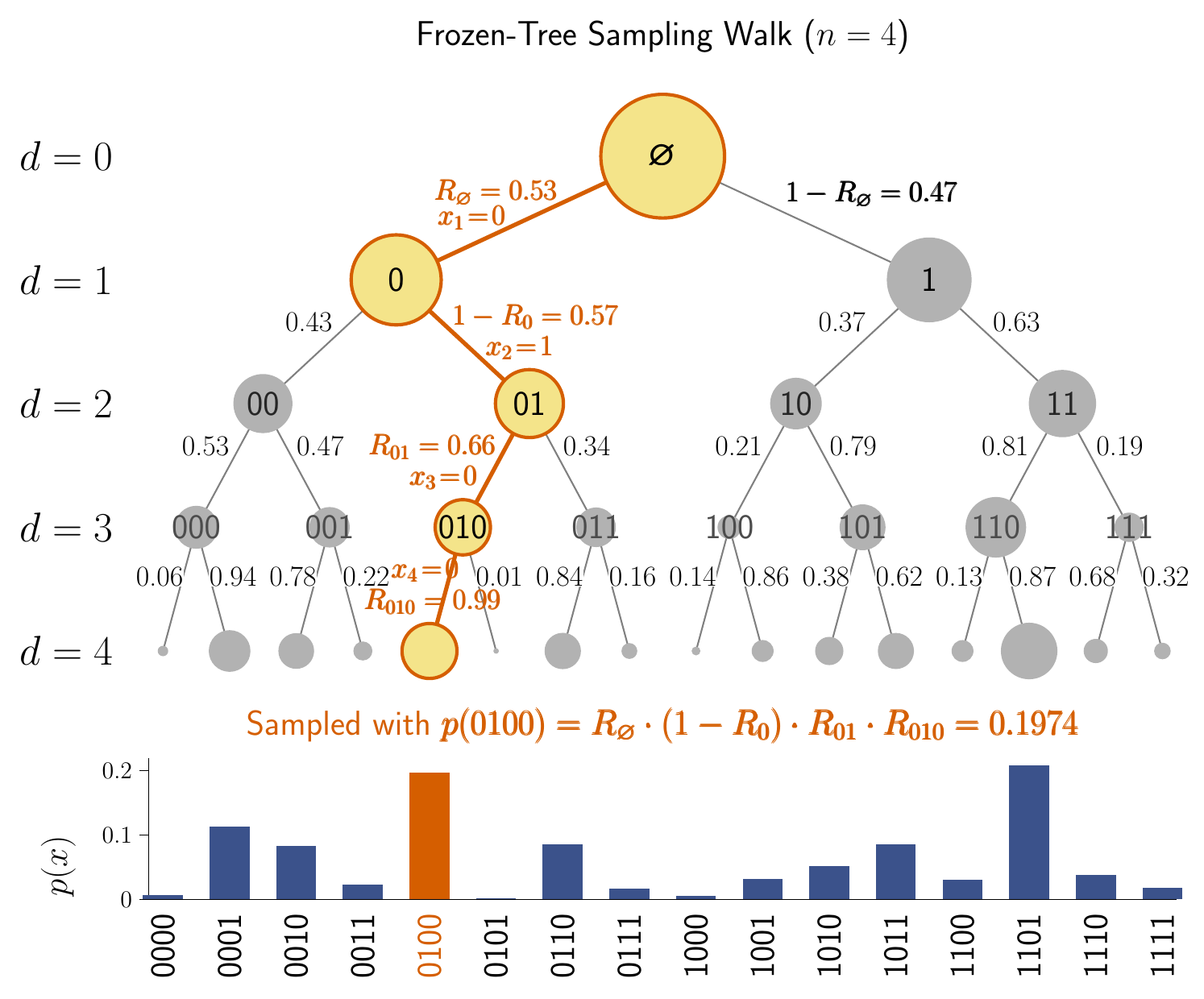}
\caption{Random walks on a frozen tree sample bitstrings from 
$p_F(x)=\abs{\braket{x}{\psi}}^2=\prod_{k=1}^n R_{x<k}^{1-x_k}(1-R_{x<k})^{x_k}$, illustrated here for 
$n=4$ qubits. During the random walk from the root to a leaf, the path at each node moves left with probability 
$R_{x<k}$ (setting  $x_k=0$) or right with probability $1-R_{x<k}$ (setting $x_k=1$). When a node is visited 
for the first time, its branch ratio $R_{x<k}$ is drawn from the Beta distribution, $\text{Beta}(K,K)$ with 
$K=2^{n-d-1}$ and frozen for later revisits. Near the root the branch ratios are close to $1/2$, whereas near 
leaves they fluctuate strongly, as indicated by the node circle sizes. For example, the probability of finding 
$x=0100$ is given by
$p(x=0100) = R_\varnothing\cdot (1-R_0)\cdot R_{01}\cdot R_{010}
= 0.53\times 0.57\times 0.66\times 0.99 \approx 0.197$.
The probability $p_F(x)$ is plotted below.
\label{Fig:2}}
\end{figure}

\paragraph*{Frozen-Tree Sampling Algorithm} ---
Once all split ratios $\{R_u\}$ are fixed, the leaf probability $p(x)$ is uniquely determined and corresponds 
to a single realization of a Dirichlet vector, that is, the Born probability of a random quantum state. 
We therefore propose a frozen-tree sampler that draws each split ratio $R_u$ once, stores it, and reuses 
it whenever a later sample revisits the same node. As shown in Fig.~\ref{Fig:2}, the sampler performs 
a random walk of $n$ steps from the root to a leaf, making one Bernoulli branch decision with probability 
$R_u$ and selecting the bit value $x_k$ at each node. This requires only $O(n)$ steps per sample, rather than 
the exponential resources of storing the full amplitude vector. For $M$ samples the total runtime is $O(Mn)$.
The sampler produces a genuine Dirichlet realization of a random quantum state because the recursive 
Beta rule is precisely the recursive representation of a symmetric Dirichlet vector,
$p_F(x)=\abs{\braket{x}{\psi}}^2 = \prod_{k=1}^n R_{x<k}^{1-x_k}(1-R_{x<k})^{x_k}$. 

As Fig.~\ref{Fig:2} depicts a random walk from the root to a leaf, the algorithm draws one bitstring as follows: 
(1) start at the root $u=\varnothing$ with depth $d=0$; (2) look up the branch ratio $R_u$ if the node has 
not been visited before, draw $R_u\sim\mathrm{Beta}(K,K)$ and store it; (3) sample $x_k=0$ with probability 
$R_u$ or $x_k=1$ with probability $1-R_u$; (4) move to the child node $u0$ or $u1$; (5) repeat until depth 
$d=n-1$. This yields one sample from the fixed distribution $p_F(x)$, where the subscript $F$ denotes 
a single frozen tree with all branch ratios $\{R_u\}$ fixed.  The $2^{n}-1$ branch ratios need not be 
generated in advance. Instead, we use lazy generation: the first time the sampler visits a node $u$, it 
draws $R_u$ and stores it in a dictionary, reusing it on all subsequent visits, so that only the actually 
visited nodes are ever instantiated.

The frozen tree needs one frozen split ratio $R_u$ per node, drawn once and identical every time a sample 
passes through node $u$. Storing every visited ratio incurs a memory cost that becomes prohibitive for 
large $n$. This is resolved by generating $R_u$ on demand, deterministically, from the node's identity 
(its prefix $u$) together with a global seed as
\begin{align}
u \;\xrightarrow{\ \mathrm{PRF(seed},u)\ }\; X_u \sim \mathrm{Unif}[0,1) \;\xrightarrow{\ f_K\ }\; R_u\,,
\end{align}
where the pseudo-random function (PRF) is a keyed, deterministic, stateless map from $(\text{seed},u)$ to a uniform 
word $X_u\in[0,1)$, implemented with a hash function such as SHA-256~\cite{SHA256} or a counter-based random 
number generator (CBRNG)~\cite{Salmon2011}. The map $f_K$ transforms $X_u$
into a Beta variate $R_u\sim\mathrm{Beta}(K,K)$,
\begin{align}
f_K(X_u)=\begin{cases}
\dfrac{1}{2}, & K>2^{102}\\[6pt]
\dfrac{1}{2}+\sigma_K\,\Phi^{-1}(X_u), & 2^{9}\le K\le 2^{102}\\[6pt]
\dfrac{\Gamma_K}{\Gamma_K+\Gamma_K'}, & K<2^{9}
\end{cases}\,,
\end{align}
where $\Phi^{-1}$ is the inverse standard-normal CDF, mapping a uniform variate to a Gaussian one.
Here $\Gamma_K,\Gamma_K'$ are independent $\mathrm{Gamma}(K,1)$ variates. The three regimes exploit 
the depth dependence of the branch-ratio spread $\sigma_K=1/\sqrt{4(2K+1)}$. Near the root $(K>2^{102}$, 
i.e., $n-d>103$), $\sigma_K$ falls below the FP64 machine epsilon, $\sigma_K<2^{-52}$, so $R_u=1/2$ to 
machine precision. No PRF call is needed and the cost is $O(1)$.  In the intermediate regime 
$(2^{9}\le K\le 2^{102}$, i.e., $10< n-d\le 103$), the Beta distribution is Gaussian to high accuracy, 
and a single inverse-CDF evaluation yields $R_u$ at $O(1)$ cost. Near the leaves ($K<2^{9}$, i.e., $n-d\le 10$), 
$R_u$ is drawn from the exact Gamma ratio at $O(K)$ cost, or via Cheng's acceptance-rejection 
algorithm~\cite{Cheng1978} at $O(1)$ cost.

\begin{figure}[t]
\includegraphics[width=0.45\textwidth]{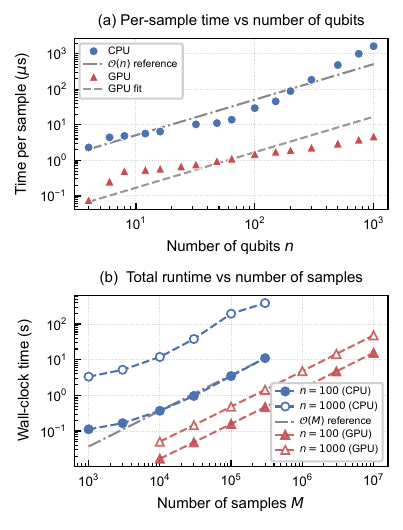}
\caption{Scaling of the frozen-tree sampler on an Intel CPU i7-3770 and an NVIDIA GPU T400. 
(a) The time per sample scales as $O(n)$ with the number of qubits $n$ and  
(b) the wall-clock time scales $O(Mn)$ with the number of samples $M$ for $n=100$ and $1,000$.
\label{Fig:scaling}}
\end{figure}

Fig.~\ref{Fig:scaling} shows the $O(n)$ scaling per sample of the frozen-tree sample code~\cite{Oh_Github}. 
It samples $M=10^7$ bitstrings of an $n=1000$-qubit random state on the GPU in a few seconds. For comparison, 
the Google Sycamore processor required about 200 seconds to collect $\sim\!10^6$ samples at 
$n=53$~\cite{Arute2019,Martinis2022}. The frozen-tree sampler is fast, exact, and noiseless. 
Fig.~\ref{Fig:scaling} (a) confirms the predicted $O(n)$ per-sample time on CPU and GPU.
Fig.~\ref{Fig:scaling} (b) shows the total runtime scaling as $O(Mn)$.

The frozen tree samples bitstrings from the probability $p_F(x)$ of a random state $\ket{\psi}$
characterized by its branching ratios $\{R_u\}$,
\begin{align}
p_F(x)=\abs{\braket{x}{\psi}}^2=\prod_{k=1}^n R_{x_{<k}}^{1-x_k}\bigl(1-R_{x_{<k}}\bigr)^{x_k}\,,
\end{align}
while a random circuit $C$ samples from the probability $p_C(x)$ of another random state 
$\ket{\varphi}=C\ket{0^n}$,
\begin{align}
p_C(x)=\abs{\bra{x}C\ket{0^n}}^2\,.
\end{align}
Both $p_F(x)$ and $p_C(x)$ are independent instances of the flat Dirichlet distribution on the $(N-1)$-simplex 
that any random state must obey. They are therefore statistically identical, $p_F(x)\stackrel{d}{=}p_C(x)$, 
even though $p_F(x)\ne p_C(x)$ in general. The only remaining unknown is the map between a circuit 
$C$ and the branching ratios $\{R_u\}$. Consequently, insofar as RCS is regarded as the task of drawing bitstrings 
from the Dirichlet distribution of a random state, the classical frozen-tree sampler, at $O(n)$
cost per sample, reproduces all statistical properties used to certify RCS. Thus no such statistic can 
serve as evidence of quantum advantage in RCS.

\begin{figure}[t]
\includegraphics[width=0.45\textwidth]{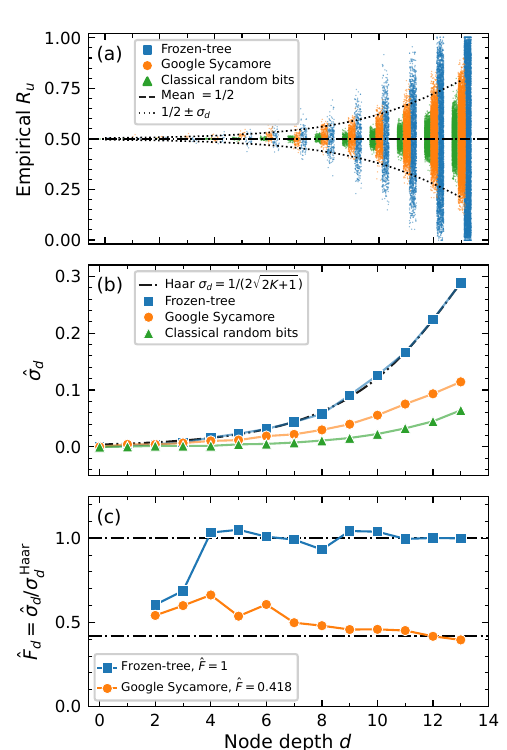}
\caption{(a) Empirical branch ratio $\hat{R}_u$, (b) its standard deviation $\hat{\sigma}_d$, and (c) 
branch-ratio fidelity $\hat{F}_d$ as functions of node depth $d$, for the frozen-tree sample, 
the Google Sycamore sample~\cite{Martinis2022}, and a classical uniform-random-bit sample. 
The number of qubits is $n=14$ and each sample contains $M=500{,}000$ bitstrings.
\label{Fig:branch_ratio}}
\end{figure}

\begin{figure}[t]
\includegraphics[width=0.45\textwidth]{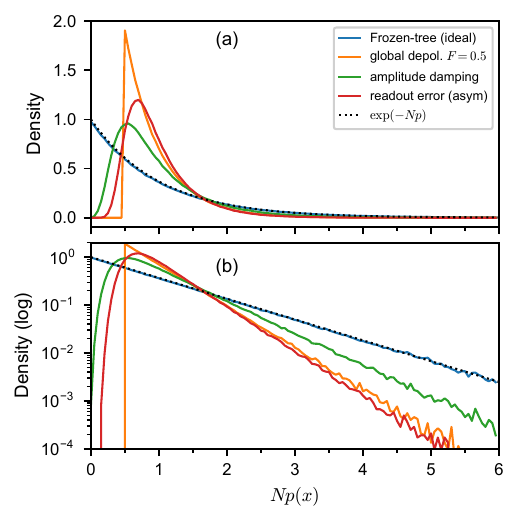}
\caption{(a) Distributions of bitstring probabilities and (b) the same distributions on semi-log scale,
plotted against $Np(x)$, obtained with the noisy frozen-tree sampler under global depolarizing noise, 
amplitude damping, and readout error. The number of qubits is $n=20$. The probabilities $p(x)$ over  
all $2^{20}$ leaves are computed directly from the noisy frozen-tree sampler, rather than estimating 
as empirical frequencies $\hat{p}(x)$ from finite samples. 
\label{Fig5}}
\end{figure}

\paragraph*{Noisy Frozen-Tree Sampler}--- The frozen-tree sampler readily accommodates various noise 
channels: global depolarizing noise, local depolarizing noise, amplitude damping, and readout error,
at the same $O(n)$ cost, with each channel modeled at a distinct stage of the tree. Global depolarizing 
noise is an affine transformation of the Dirichlet vector~\cite{Oh2026}: the noisy leaf probability is 
a mixture of the ideal probability and the uniform distribution,
\begin{align}
\widehat{p}(x) = F\,p_{\text{ideal}}(x) + \frac{1-F}{2^n}\,,
\end{align}
where $F$ is the circuit fidelity. This is exact at any $n$ and requires no noise-trajectory sampling.

Readout error and amplitude damping act qubit-by-qubit on the completed bitstring, after the tree path 
is fixed. A true string $x_k$ is corrupted into an observed string $y_k$ through per-qubit confusion 
matrices ${C_k(y_k \mid x_k)}$, so the observed-string probability is
\begin{align}
\widehat{p}(y) = \sum_{x} p_{\text{ideal}}(x)\prod_{k=1}^n C_k(y_k\mid x_k)\,.
\label{Eq:noisy_tree}
\end{align}
For readout error, the confusion matrix is 
\begin{align}
C_k = \begin{pmatrix} 
      P(0\mid 0) & P(0\mid 1)\\[4pt]
      P(1\mid 0) & P(1\mid  1)
      \end{pmatrix}
    = \begin{pmatrix} 
      1-\epsilon^{01}_k & \epsilon^{10}_k\\[4pt]
      \epsilon^{01}_k   & 1-\epsilon^{10}_k
      \end{pmatrix}\,,
\end{align}
where $P(0\mid 1) =\epsilon^{10}_k$ ($P(1\mid 0) =\epsilon^{01}_k$) is the probability that 1 (0) of the $k$-th qubit 
is incorrectly measured as 0 (1).  Amplitude damping is a one-sided decay from $\ket{1}$ to $\ket{0}$ at rate 
$\gamma_k = 1-\exp(-t_{\text{meas}}/T_{1,k})$, where $T_{1,k}$  is the relaxation time of qubit $k$
and $t_\text{meas}$ is the measurement window. Its confusion matrix is
\begin{align}
C_k =\begin{pmatrix} 
     1 & \gamma_k \\ 
     0 & 1-\gamma_k
     \end{pmatrix}\,.
\end{align}
The exact evaluation of Eq.~(\ref{Eq:noisy_tree}) is feasible only at enumerable sizes $(n\lesssim 25)$. 
For larger $n$, the confusion matrices are instead folded into each conditional branch probability, 
preserving the $O(n)$ cost.

Fig.~\ref{Fig:branch_ratio} shows how the branch ratio $R_u$ fluctuates with node depth $d$ for the 
frozen tree, the Google Sycamore data~\cite{Martinis2022}, and a classical uniform-random-bit sample. 
Fig.~\ref{Fig:branch_ratio} (b) confirms that the frozen-tree spread follows
$\sigma_d=1/(2\sqrt{2K+1})$ exactly, whereas Sycamore lies between ideal RCS and classical random bits, 
with $R_u$ fluctuating more strongly toward the leaves. Fig.~\ref{Fig:branch_ratio} (c) shows 
the frozen tree gives $\hat{F}=1$ across all depths, while Sycamore saturates at $\hat{F}=0.418$.

Figure~\ref{Fig5} plots the distributions $\mathrm{Pr}(p)$ of bitstring probabilities, scaled as $Np$, 
for $n=20$ under global depolarizing noise $(F=0.5)$, amplitude damping $(\gamma_k=0.05)$, and 
readout error ($\epsilon^{01}=0.02, \epsilon^{10}=0.06)$. A notable feature of the frozen-tree 
is that these noisy distributions are obtained by enumerating all $2^n$ leaves and 
applying each channel operator, \textit{without sampling bitstrings}. Because $p(x)$
is given in closed form as a product of branch ratios, the full leaf-probability vector $\pmb{p}=(p(x))$ 
is computed exactly, and each physical noise channel acts as a linear map on it. This contrasts with 
the conventional analysis of RCS, in which the noisy output distribution can only be estimated by 
drawing many bitstrings and histogramming them with estimation error. Like $\mathrm{Pr}(p)$,
every sample-based verification tool, such as cross-entropy, heavy-output generation, 
and anticoncentration, is a functional of $p$ and is therefore exactly determined rather than 
empirically estimated. As shown in Fig.~\ref{Fig5}, amplitude damping and readout error suppress 
the exponential peak and deplete the small $Np$ region relative to the ideal law $\exp(-Np)$, 
whereas global depolarizing noise leaves the exponential shape intact but rescaled and shifted 
to the right by the fidelity $F$~\cite{Oh2026}. Amplitude damping (green) retains a heavier tail than 
either depolarizing noise or readout error, remaining well above them and falling off only near 
$Np\approx6-7$; this is the signature of a one-sided channel, which concentrates probability and preserves 
more of the exponential tail. Depolarizing noise (orange) and readout error (red) decay faster, essentially 
vanishing by $Np\approx 4-5$, since they move probability mass away from the high $Np$ leaves. These are 
the same statistical distortions seen in the Google Sycamore data~\cite{Oh2026}. 
The frozen-tree sampler thus reproduces not only the ideal exponential distribution but also 
its realistic noisy statistics exactly. This implies that no test acting on samples alone can distinguish 
the frozen-tree sampler from a quantum random circuit.

\paragraph*{Summary and Discussion}--- We have shown that the probability 
$p(x)=\abs{\braket{x}{\psi}}^2$ of finding bitstring $x$ in a random state $\ket{\psi}$ is 
represented by a binary tree whose branching ratios $R_u$ follow the Beta distribution $\mathrm{Beta}(K,K)$, 
a direct consequence of the exact conditional scale invariance of random states~\cite{Oh2026,Oh2026b}. 
Through random walks from the root to leaves, the frozen-tree sampler 
draws bitstrings from the Born distribution 
of an $n$-qubit random state in $O(n)$ time per sample. It generates $10^7$ samples at $n=1000$ within 
seconds on a personal computer. It accommodates depolarizing noise, amplitude damping, and readout error 
at the same cost. 

The probability vector $\boldsymbol{p}=(p(x))$ of any random state is a flat Dirichlet vector on 
the simplex. A random circuit produces one such vector and the frozen-tree sampler produces another; 
the two are independent realizations of the same distribution, and statistically identical 
$p_F(x)\stackrel{d}{=}p_C(x)$ while $p_F(x)\ne p_C(x)$. It follows that any verification method
built from the leaf-probability $p$ such as the exponential distribution, anticoncentration, linear XEB, 
and heavy-output generation, takes the same value for the frozen-tree sampler as for a random quantum 
circuit. Since the frozen tree reproduces all of these classically in $O(n)$ time, none of them can 
by itself certify quantum advantage. They are necessary features of the Dirichlet ensemble that 
a classical sampler shares, not sufficient evidence of hardness. 

Any surviving quantum advantage of RCS must therefore reside in the circuit-specific realization 
$p_C(x)$, the map from a given circuit to its amplitudes, rather than in the sampled statistics. 
Whether even this residual survives is the subject of our companion work~\cite{Oh2026c}, which 
constructs a Hurwitz–frozen-tree circuit $C_F$ implementing $p_{C_F}(x)=\abs{\bra{x}C_F\ket{0^n}}^2=p_F(x)$: 
a quantum circuit whose output is, by construction, classically sampled in $O(n)$ time. 
Such a family shows that a circuit can pass the RCS benchmark while admitting an efficient classical sampler, 
underscoring that the benchmark's statistics do not witness advantage. Taken together, these results argue 
that random circuit sampling, as currently verified, is not a sound stand-alone benchmark for quantum advantage.

\bibliographystyle{apsrev4-2}
\bibliography{Frozen_Tree.bib}
\vfill
\vspace{3cm}
\break\newpage
\appendix*
\section{Appendix: Verification statistics as functionals of the leaf law}
The exact-enumeration results of the main text are limited to $n\lesssim 25$, where the full $2^n$ probability 
vector can be constructed. To confirm that the frozen-tree sampler reproduces the exponential leaf statistics 
in the regime relevant to hardware experiments, we draw $10^{6}$ leaves uniformly from the $2^{n}$ possible
bitstrings at $n=30$, $40$, and $50$, and evaluate each leaf probability $p(x)$ in closed form from 
its branch ratios. Because $p(x)$ is computed directly rather than estimated from finite-frequency counts, 
no sampling error enters the individual probabilities; the only randomness is in which leaves are drawn. 
As shown in Fig.~\ref{Fig:randleaf}, the histograms of the rescaled variable $Np(x)$ coincide with the ideal law
$\exp(-Np)$ at every size, even though $10^{6}$ leaves represent a vanishing fraction---from $10^{-3}$ at $n=30$ 
down to $10^{-9}$ at $n=50$---of the full distribution. The distribution $\Pr(Np)$ shown here is built from 
this finite set of $10^{6}$ computed values of $p(x)$, far smaller than the total number of leaves $2^{n}$
($10^{6}\ll 2^{50}\approx 10^{15}$ at $n=50$), yet the exponential law is fully recovered. The leaf law is 
therefore a property of the flat Dirichlet ensemble, recovered by the $O(n)$ sampler at scales where
direct enumeration is infeasible.

\begin{figure}[h]
\centering
\includegraphics[width=\columnwidth]{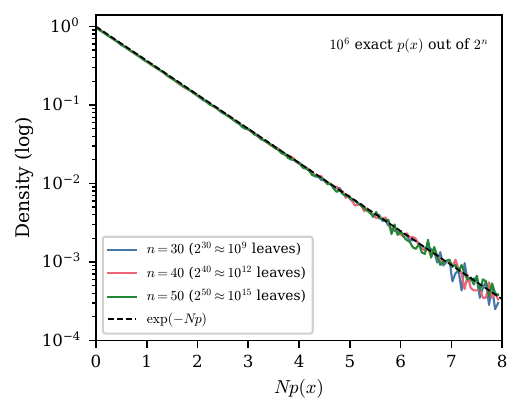}
\caption{Validation that the frozen-tree sampler reproduces the exponential law $\exp(-Np)$ at large $n$
where the full $2^n$ amplitude vector cannot be enumerated. For $n=30$, $40$, and $50$, the sampler draws
$10^{6}$ random leaves and evaluates each leaf probability $p(x)$ in closed form from its branch ratios; the
histograms of the rescaled variable $Np(x)$ coincide with the ideal exponential law (dashed) at every $n$.
The distribution is thus $n$-independent in $Np$: the mean is $\langle Np\rangle = 1.00$ and the
high-probability tail fraction $\Pr(Np>4)\approx 0.018$ for all three
sizes, confirming that the sampled statistics track the flat
Dirichlet leaf law rather than the Hilbert-space dimension.}
\label{Fig:randleaf}
\end{figure}
\end{document}